\documentclass[a4paper, 11pt]{amsart}
\usepackage{amssymb,array}
\usepackage{url}
\usepackage{graphicx}
\usepackage{subfigure}
\usepackage{hyperref}

\renewcommand{\leq}{\leqslant}
\renewcommand{\geq}{\geqslant}

\newcommand{\R}{\mathbb{R}}
\newcommand{\C}{\mathbb{C}}

\newcommand{\norm}[1]{\left\Vert #1\right\Vert}
\newcommand{\normt}[1]{\left\Vert #1\right\Vert_{(t)}}

\newcommand{\Bell}{\mathrm{Bell}}

\DeclareMathOperator{\Gr}{Gr}

\newcommand{\scalar}[2]{\langle #1 , #2\rangle}

\renewcommand{\phi}{\varphi}
\newcommand{\iy}{\infty}

\newcommand{\xopt}{x^*_t}

\newtheorem{theorem}{Theorem}[section]
\newtheorem{definition}[theorem]{Definition}
\newtheorem{proposition}[theorem]{Proposition}

\newtheorem{lemma}[theorem]{Lemma}

\newtheorem{corollary}[theorem]{Corollary}

\begin{document}

\title[Almost one bit violation]{Almost one bit 
violation for the additivity of the minimum output
entropy}
\author{Serban T. Belinschi}
\address{Department of Mathematics \& Statistics,
Queen's University, and Institute of Mathematics ``Simion Stoilow'' of the Romanian Academy;
Jeffrey Hall,
Kingston, ON K7L 3N6 CANADA}
\email{ sbelinsch@mast.queensu.ca}
\author{Beno\^{\i}t Collins}
\address{
D\'epartement de Math\'ematique et Statistique, Universit\'e d'Ottawa,
585 King Edward, Ottawa, ON, K1N6N5 Canada
and
WPI AIMR Tohoku University, Mathematics Unit
and
CNRS, Institut Camille Jordan Universit\'e  Lyon 1}
\email{bcollins@uottawa.ca}
\author{Ion Nechita}
\address{CNRS, Laboratoire de Physique Th\'eorique, IRSAMC, Universit\'e de Toulouse, UPS, 31062 Toulouse, France}
\email{nechita@irsamc.ups-tlse.fr}

\begin{abstract}
In a previous paper,
we proved that the limit of the collection
of possible eigenvalues of output states of a random  quantum channel is a deterministic, compact set $K_{k,t}$.
We also showed that the set $K_{k,t}$ is obtained, up to an intersection, as the unit ball of the
dual of a free compression norm.

In this paper, we identify the maximum of $\ell^p$ norms on the set $K_{k,t}$ and
prove that the maximum is attained on a vector of shape $(a,b,\ldots ,b)$ where $a > b$.
In particular, we compute the precise limit value of the minimum output entropy of a single random quantum channel.
As a corollary, we show that for any $\varepsilon >0$, it is possible to obtain a
violation for the additivity of the minimum output entropy for an output dimension as low as
$183$, and that for appropriate choice of parameters, the violation can be as large as
$\log 2 -\varepsilon$. Conversely, our result implies that, with probability one,
one does not obtain a violation of additivity using conjugate random quantum channels and the Bell state, in dimension $182$ and less.
\end{abstract}

\maketitle

\section{Introduction}

Let $(\mathcal A,\tau)$ be a tracial von Neumann non-commutative probability space.
On this vector space, given $t\in (0,1)$,
let us introduce the quantity
$||x||_{(t)}=||pxp||$ where $p$ is a projection of normalized trace $t$, free from $x$.
As we indicate in Section \ref{sec:t-norm}, this is a norm, which we call the $(t)$-norm.

In this paper we are interested in the $(t)$-norm
restricted to subalgebras of $(\mathcal A,\tau)$ of the form $(\mathbb{C}^k,k^{-1}\sum\delta_i)$
(generated by $k$ selfadjoint orthogonal projections of trace $1/k$).
The set $K_{k,t}$ is the dual of the unit ball for the $(t)$-norm, intersected with the
$(k-1)$-dimensional
probability simplex
 $\Delta_k = \{y \in \R_+^k \; | \; \sum_{i=1}^k y_i = 1\}$.

This set was introduced in \cite{bcn1} and we recall some of its properties in Section
\ref{subsec:def-kkt}.
The interest of $K_{k,t}$ is that it describes the limit of the collection of
\emph{all possible outputs} of
states (or eigenvalues thereof)  in the large dimension limit for
a natural family of random quantum channels (see Section \ref{sec:large-asymptotics}).

In this paper, we state and
study a maximization problem of $\ell^p$ norms on
$K_{k,t}$.
Our main result (stated below as Theorem \ref{thm:minimum-pnorm-Kkt}) is that the maximum is reached on a point that we call $\xopt$:
\begin{theorem}
The maximum of the $\ell^p$ norm on $K_{k,t}$ is reached at the vector $\xopt = (a,b,\ldots ,b)$, with $a>b$  depending only on $k$ and $t$. In particular, the point where the maximum is achieved does not depend on $p$.
\end{theorem}

With this result, we are in position to supply the optimal bounds for the
random techniques at hand in order to disprove the additivity of the
minimum output entropy (MOE). 
Our main application can be summarized as follows.

\begin{theorem}
Violations of the additivity of the MOE, using conjugate random quantum channels and the Bell state, can occur iff the output space has dimension
 at least $183$. Almost surely, the defect of additivity is less than $\log 2$,
 and it can be made as close as desired to  $\log 2$.
 \end{theorem}
The detailed version corresponding to the above is Theorem
\ref{thm:violation} and its proof.

This theorem completely solves the problem of computing the MOE for single random quantum channels. It fully clarifies and optimizes the extent to which all available techniques so far
in the problem of the additivity of the MOE can give violation of additivity.

Our paper is organized as follows.
In Section
\ref{sec:main}
we state a minimization problem arising from free probability.
Section
\ref{sec:proof} is devoted to its proof.
In Sections
\ref{sec:MOE}
and \ref{sec:limiting},
we show that this minimization problem translates into the
computation of the MOE for random quantum channels
and in Section
\ref{sec:violation}
we use this to obtain new bounds for the violation
of the MOE, which are optimal in some sense.

\bigskip

\noindent \textit{Acknowledgments.}
The authors would like to thank Motohisa Fukuda for inspiring discussions. The authors had opportunities to meet in pairs at Saarbr\"ucken, Toulouse, Lyon and Ottawa to complete their research, and thank these institutions for a fruitful working environment. One of the visits has
been supported in part by an MAE/MESR/DAAD ``Procope'' joint travel grant between Universit\'e
de Toulouse ``Paul Sabatier'' and Universit\"{a}t de Saarlandes.
SB and BC's research was supported by NSERC discovery grants.
SB's research was also partly supported by a RIG from Queen's University.
BC's research was supported by an ERA and in part by AIMR. IN's research has been supported by the ANR projects {OSQPI} {2011 BS01
  008 01} and {RMTQIT}  {ANR-12-IS01-0001-01}, as well as by the PEPS-ICQ CNRS project \mbox{Cogit}. IN also acknowledges the hospitality of the Technische Universit\"at M\"unchen, where part of this research was conducted.

\section{Definitions and statement of the main result}
\label{sec:main}

\subsection{The $(t)$-norm}
\label{sec:t-norm}

\begin{definition}\label{def:t-norm}
For a positive integer $k$, embed $\R^k$ as a selfadjoint real subalgebra $\mathcal R$ of
a $\mathrm{II}_1$ factor  $\mathcal A$, spanned by $k$ mutually orthogonal projections of normalized trace $1/k$. Let $p_t$ be a projection of rank $t \in (0,1]$ in $\mathcal A$,
free from $\mathcal R$. On the real vector space $\R^k$, we introduce the following quantity,
called the \emph{$(t)$-norm}:
\begin{equation}
        \normt{x}:=\norm{p_t x p_t}_{\infty},
\end{equation}
where the vector $x \in \R^k$ is identified with its image in $\mathcal R$.
\end{definition}

In the sequel, the notions of 
$\mathrm{II}_1$ factor and freeness do not matter.
We refer the interested reader
to our previous paper \cite{bcn1} and to reference texts
\cite{voiculescu-dykema-nica,nsp} for detail.
For the purpose of this paper, it is enough to know that even though it is difficult to
compute explicitly the $(t)$-norm, there is a simple algebraic definition of it, given
in the proposition \ref{technical-proposition} below.

We make use of the following notation:
\begin{equation}
        (1^j0^{k-j}) = (\underbrace{1, 1, \ldots, 1}_{j \text{ times}},\underbrace{0, 0, \ldots, 0}_{k-j \text{ times}}) \in \R^k,
\end{equation}
and $1^k = (1^k0^0)$. We denote by

\begin{equation}
	G_\mu(z) = \int \frac{d\mu(t)}{z-t}
\end{equation}
the Cauchy-Stieltjes transform of the measure $\mu$ and by $F_\mu = 1/G_\mu$. For $x = (x_1, \ldots, x_k) \in \mathbb R^k$, we write
\begin{equation}
\mu_x = \frac 1 k \sum_{i=1}^k \delta_{x_i}.	
\end{equation}

\begin{proposition}\label{technical-proposition}
The quantity $\normt{\cdot}$ has the following properties:
        \begin{enumerate}
        \item  It is indeed a norm.

\item
        It is invariant under permutation of coordinates
                \begin{equation*}
                        \normt{(x_1, x_2, \ldots, x_k)} = \normt{(x_{\sigma(1)}, x_{\sigma(2)}, \ldots, x_{\sigma(k)})} \qquad \forall \sigma \in S_k.
                \end{equation*}
        \item For any $x\in\mathbb R^k$,
                \begin{equation*}
                        \frac1t\normt{x} = \frac1tw_x+\left(1-\frac1t\right)F_{\mu_x}(w_x),
                \end{equation*}
where $w_x$ is the largest in absolute value solution to the equation
\begin{equation*}
F_{\mu_x}(w)\left(F_{\mu_x}'(w)-\frac{1}{1-t}\right)=0.
\end{equation*}
        \item The function $\mathbb R\ni s\mapsto\normt{x+s1^k}$ has a unique
point of minimum $s_0$, with the property that
$$
\frac{1}{t}\normt{x+s_01^k}=\max\mathrm{supp}(\mu_{x+s_01^k}^{\boxplus 1/t})=
-\min\mathrm{supp}(\mu_{x+s_01^k}^{\boxplus 1/t}).
$$
Moreover,
$$
\normt{x+s1^k}=|s-s_0|+\normt{x+s_01^k},\quad s\in\mathbb R.
$$
        \item For all $j=1, 2, \ldots, k$, one has
                \begin{equation*}
                        \normt{(1^j0^{k-j})}=\phi(u,t)=
                        \begin{cases}
                                t +u-2tu +2\sqrt{tu (1-t)(1-u)}  &\text{ if } t + u < 1,\\
                                1  &\text{ if } t + u \geq 1,
                        \end{cases}
                \end{equation*}
                where $u = j/k$. 
        \item If $x\in\mathbb R_+^k$, the $m$ largest coordinates of $x$ are all equal and  $m/k+t\geq 1$,
        then $||x||_{(t)}=||x||_\infty$
        \end{enumerate}
\end{proposition}

\subsection{Definition of the convex body $K_{k,t}$}
\label{subsec:def-kkt}

We introduce now the convex body $K_{k,t}\subset \Delta_k$ as follows:
\begin{equation}\label{eq:convex}
        K_{k,t}:=\{ \lambda \in\Delta_{k} \; |\; \forall a\in\Delta_{k} , \scalar{\lambda}{a} \leq \normt a
        \},
\end{equation}
where $\scalar{\cdot}{\cdot}$ denotes the canonical scalar product in $\R^k$.
\begin{lemma}\label{max}
For any $\lambda \in K_{k,t}$, we have
$$
\max_{a\in\Delta_k}\scalar{\lambda}{a}-\normt{a}=\max_{a\in\mathbb R^k}\scalar{\lambda}{a}-\normt{a}\leq0.
$$
In other words, $K_{k,t}$ is the intersection of the probability simplex $\Delta_k$ with the
unit ball of the dual norm of $\normt{\cdot}$.
\end{lemma}
\begin{proof}
Fix $\lambda\in K_{k,t}$ and $a_0\in\mathbb R^k$. By Proposition \ref{technical-proposition},
it follows that $\normt{a_0+s1^k}=|s-s_0|+\min_{r\in
\mathbb R}\normt{a_0+r1^k}$, where $s_0$ is exactly the point where this minimum is
reached.
We have
\begin{eqnarray*}
\scalar{\lambda}{a+s1^k}-\normt{a+s1^k} & = & \scalar{\lambda}{a}+s-(|s-s_0|+
\normt{a+s_01^k})\\
& = & \scalar{\lambda}{a+s_01^k}-\normt{a+s_01^k}\\
& & \mbox{}-|s-s_0|+(s-s_0)\\
& \leq & \scalar{\lambda}{a+s_01^k}-\normt{a+s_01^k},
\end{eqnarray*}
with equality if and only if $s\ge s_0$.
Thus,
$\scalar{\lambda}{a+\|a\|_\infty1^k}-\normt{a+\|a\|_\infty1^k}=\max_{s\in\mathbb R}
\scalar{\lambda}{a+s1^k}-\normt{a+s1^k},$ and $a+\|a\|_\infty1^k\in\mathbb R_+^k$.
This proves the fact that
$$
\max_{a\in\mathbb R^k}\scalar{\lambda}{a}-\normt{a}=\max_{a\in\mathbb R^k_+}\scalar{\lambda}{a}-\normt{a}.
$$
We have assumed that $\lambda\in K_{k,t}$, so that $\max_{a\in\Delta_k}\scalar{\lambda}{a}-\normt{a}\leq0$. Thus, $\max_{a\in\mathbb R^k_+}\scalar{\lambda}{a}-\normt{a}$
cannot be a positive number or $+\infty$. Indeed, if this were not the case, then for any
$a\in\mathbb R^k_+$ so that $\scalar{\lambda}{a}-\normt{a}>0$ we can take
$\scalar{\lambda}{a/\|a\|_1}-\normt{a/\|a\|_1}=(\scalar{\lambda}{a}-\normt{a})/\|a\|_1$
to obtain a contradiction with the condition $\lambda\in K_{k,t}$. On the other hand,
$\scalar{\lambda}{1^k}=\normt{1^k}=1$, so the equality is reached, and the above maximum is
indeed zero (see also last section of \cite{bcn1}).
\end{proof}

\subsection{Main result}

The main result of this paper is that the maximum of $\ell^p$ norm 
on $K_{k,t}$ is reached at a precise point (up to permutation of coordinates), to be identified below.
Moreover, this point does not depend on the value of $p$. The value of this maximum will be easily computed. Since most of the properties
we prove for vectors in $\Delta_k$ do not depend on the order of the coordinate entries
of those vectors, we shall often focus our attention on the subset
$$
\Delta_k^\downarrow:=\{x\in\Delta_k\colon x=(x_1\ge x_2\ge\cdots\ge x_k)\}.
$$

Recall that $e_1 = (1, 0, \ldots, 0) \in \R^k$ and let
\begin{equation}\label{eq:def-xopt}
        \xopt=\left(\normt{e_1}, \underbrace{\frac{1-\normt{e_1}}{k-1}, \ldots, \frac{1-\normt{e_1}}{k-1}}_{k-1 \text{ times}}\right).
\end{equation}

With this notation we are able to state our main result

\begin{theorem}\label{thm:minimum-pnorm-Kkt}
For any $p> 1$, the maximum of the $\ell^p$ norm on $K_{k,t}$ is reached at the point
$\xopt$.

\end{theorem}

The next section is devoted to the proof of this result.

\section{Proof of Theorem \ref{thm:minimum-pnorm-Kkt}}
\label{sec:proof}

\subsection{Strategy of the proof}

Our proof relies on the following crucial observation (see also \cite[Lemma 6.1]{bcn1}):
\begin{lemma}\label{Lemma3.1}
The convex body $K_{k,t}$ is the
image of $\Delta_k$ via the subdifferential of the $(t)$-norm $\normt{\cdot}$:
$$
K_{k,t}=\bigcup_{x\in\Delta_k}\left(\partial\normt{x}\right)\cap\Delta_k.
$$
This correspondence between $K_{k,t}$ and $\partial\normt{\cdot}$ has the following properties:
\begin{enumerate}
\item If $\normt{\cdot}$ is differentiable in $x\in\Delta_k$, then 
$\nabla\normt{x}\in K_{k,t}$.
\item In addition, for any $t<1-\frac1k$, the set of points of differentiability of $\normt{\cdot}$
is dense in $\Delta_{k}$.
\item The map $x\to \nabla ||x||_{(t)}$ is increasing, in the sense that
$\scalar{x-y}{\nabla\normt{x}-\nabla\normt{y}}\ge0$.
\item Let $x=(x_1\ge x_2\ge\cdots\ge x_k)\in\Delta_k^\downarrow$. If $x_j>x_{j+1}$ for a $j\in\{1,\dots,k-
1\}$ and $\lambda\in\partial\normt{x}$, then
$\lambda_j\ge\lambda_{j+1}$. In particular, the correspondence $x\mapsto\nabla\normt{x}$
preserves monotonicity of vector coordinates.
\end{enumerate}
\end{lemma}

\begin{proof}
We shall prove the main statement of our lemma by double inclusion. By definition,
$$
\partial\normt{x}=\left\{\lambda\in\mathbb R^k\; |\; \forall b\in\mathbb R^k , \scalar{\lambda}{b-
x} \leq \normt{b}-\normt{x}\right\}.
$$
Let now $\lambda\in K_{k,t}$, which, by definition means that $\lambda\in\Delta_k$ and
$\scalar{\lambda}{a}\leq\normt{a}$ for all $a\in\Delta_k$. Let $a_0\in\Delta_k$ be so that
$\scalar{\lambda}{a_0}-\normt{a_0}=\max_{a\in\Delta_k}\scalar{\lambda}{a}-\normt{a}=
\max_{a\in\mathbb R^k}\scalar{\lambda}{a}-\normt{a}$ (by Lemma \ref{max}). We claim that $\lambda\in
\partial\normt{a_0}$. Indeed, for an arbitrary $b\in\mathbb R^k$,
$$
\scalar{\lambda}{b-a_0}\leq\normt{b}-\normt{a_0}\iff\scalar{\lambda}{b}-\normt{b}\leq
\scalar{\lambda}{a_0}-\normt{a_0}.
$$
The right hand side of the equivalence is true by the definition of $a_0$, while the left
hand side is the condition in the the definition of $\partial\normt{a_0}$. Thus ``$\subseteq$'' is proved.

To prove ``$\supseteq$'', choose $x\in\Delta_k$. If $\lambda\in\partial\normt{x}\cap\Delta_k$, then by
definition $\scalar{\lambda}{b}-\scalar{\lambda}{x}\leq\normt{b}-\normt{x}$ for all $b\in
\mathbb R^k$. Then, for given $a\in\Delta_k$ we choose $b=a+x$ to conclude that
$\scalar{\lambda}{a}\leq\normt{a+x}-\normt{x}$. By the triangle inequality
$\normt{a+x}-\normt{x}\leq\normt{a}$, which gives us $\scalar{\lambda}{a}\leq\normt{a}.$
Since this holds for any arbitrary $a\in\Delta_k$, by the definition of $K_{k,t}$ we obtain that
$\lambda\in K_{k,t}$. This gives us the required inclusion.

Since at points of differentiability we have $\partial\normt{x}=\{\nabla\normt{x}\}$,
the proof of item (1) is complete.

We have shown in \cite[Remark 6.3]{bcn1} that the set of points of non-{differentiability} of $\normt{\cdot}$
is simply the set of points $(x_1,\dots,x_k)\in\Delta_k^\downarrow$ with the property that
$x_1=\cdots=x_m$ for an $m\ge k(1-t)$. This describes a face of dimension at most $k-m$
of $\Delta_k^\downarrow$, proving item (2).

Item (3) is a straightforward consequence of the convexity of $\normt{\cdot}$ (see
\cite[Proposition 17.10]{Bau}).

Finally, let $x=(x_1\ge x_2\ge\cdots\ge x_k)\in\Delta_k^\downarrow$.
By definition, $\lambda\in\partial\normt{x}$ iff $\forall b\in\mathbb R^k, \scalar{\lambda}{b-x}
 \leq \normt{b}-\normt{x}$. This last relation implies (by picking $b=0$ and $b=2x$) that
$\scalar{\lambda}{x}=\normt{x}$. Thus $\scalar{\lambda}{b} \leq \normt{b}$ for all $b\in\mathbb
R^k$. It is known that if $x$ has decreasing coordinates, then the set of scalar products of $x$
with all the vectors obtained by permuting $\lambda$'s coordinates will be maximized by making
$\lambda$'s coordinates also decreasing. If $\lambda$ has two coordinates in the wrong order,
we simply choose as $b$ the vector $x$ in which we have permuted two coordinates in such a
manner as to match the ones of $\lambda$. Since the $(t)$-norm is invariant under such a
permutation we obtain $\scalar{\lambda}{b} \leq \normt{b}=\normt{x}=
\scalar{\lambda}{x}<\scalar{\lambda}{b},$ an obvious contradiction. This proves (4).
\end{proof} 
Let us remark that item (4) of the above lemma is true for any convex function
invariant under permutation of coordinates, not only for the $(t)$-norm.

\subsection{Some technical results about the $(t)$-norm}

First we recall a couple of facts from the literature

\begin{proposition}[\cite{NS}]
The following holds true
$$\mu_{pxp}=(1-t)\delta_0+ tD_{t}\mu^{\boxplus 1/t}_x.$$
\end{proposition}

We conclude from this that
$$\normt{x}=\|pxp\|_\infty=t\max\textrm{supp}(\mu_x^{\boxplus 1/t}).$$

We shall frequently use the following

\noindent{\bf Notation:}
\begin{equation}\label{m_x}
m_x=\#\{j \in \{1, \ldots, k\} \colon x_j=\max_{1\le r\le k}
x_r\}.
\end{equation}

Next, we have
\begin{lemma}
Let $x\in\mathbb R^k\setminus\mathbb R1^k$.
Whenever $m_x/k+t\leq1$ (and in particular when $t\le1/k$), the quantity
$\max\textrm{supp}(\mu_x^{\boxplus 1/t})$ coincides with the largest point of
non-analyticity of $F_{\mu^{\boxplus 1/t}_x}=1/G_{\mu^{\boxplus 1/t}_x}$ along the real line,
where $G_\mu$ is the Cauchy-Stieltjes transform of the measure $\mu$.
\end{lemma}

We denote from now on $\omega_t(\normt{x}/t)=w$ (sometimes $=w(x)$, as the dependence
in $t$ will not be interesting here and we suppress it), where $\omega_t$ is the so-called
subordination function, uniquely determined by the functional equation \cite{BB-mathz, BB-imrn, Biane}
$$
\omega_t(z)=tz+(1-t)F_{\mu^{\boxplus 1/t}_x}(z)=tz+(1-t)F_{\mu_x}(\omega_t(z)).
$$
\begin{proposition}\label{prop3.4}
Let $x\in\mathbb R^k\setminus\mathbb R1^k$.
Whenever $m_x/k+t\leq1$
(and in particular when $t\le1/k$),
the following holds true:
\begin{equation}\label{1}
\left[\sum_{j=1}^k\frac{1}{w-x_j}\right]^2=
k(1-t)\sum_{j=1}^k\frac{1}{(w-x_j)^2},\quad x\in\mathbb R^k_+.
\end{equation}
\end{proposition}

\begin{proof}
Indeed, as seen just above the statement of our proposition,  
the function
 $\omega_t$ is analytic at $x$ iff $F_{\mu^{\boxplus 1/t}_x}$ is.
Now we differentiate the above:
$$
\omega'_t(z)=t+(1-t)F'_{\mu_x}(\omega_t(z))\omega_t'(z).
$$
This implies that in the point $\normt{x}/t$ where $\omega_t'$ is infinite we have
$$
1=(1-t)F'_{\mu_x}(\omega_t(\normt{x}/t))=(1-t)\frac{\frac1k\sum_{j=1}^k\frac{1}{(\omega_t(
\normt{x}/t)-x_j)^2}}{
\left[\frac1k\sum_{j=1}^k\frac{1}{\omega_t(\normt{x}/t)-x_j}\right]^2}.
$$
This completes the proof.
\end{proof}

We now state and prove a lemma regarding the position of the point $w$ with respect to $x$.

\begin{lemma}\label{lem:position-w}
For all vectors $x \in \Delta_k^\downarrow$ such that $t+m_x / k < 1$, we have that $w = w(x) >x_1$.
\end{lemma}
\begin{proof}
By definition, $w$ is the largest root of the equation $\phi(v)=0$, where
\begin{equation}
\phi(v) = (1-t) \frac{1}{k}\sum_{i=1}^k (v-x_i)^{-2} - \left[ \frac{1}{k}\sum_{i=1}^k (v-x_i)^{-1} \right]^2.
\end{equation}
We have that
\begin{equation}
\phi(x_1+\epsilon) = (1-t) \frac{1}{k} \frac{m_x}{\epsilon^2} - \left[ \frac{m_x}{k} \epsilon^{-1} \right]^2 + o(\epsilon^{-2}) = \frac{m_x}{k\epsilon^2}(1-t-m_x/k) +  o(\epsilon^{-2}),
\end{equation}
so that
\begin{equation}
\lim_{v \to x_1^+} \phi(v) = + \infty.
\end{equation}
In the same way, when $v \to \infty$, we have
\begin{equation}
\phi(v) =v^{-2} \left[ (1-t) - 1^2\right] + o(v^{-2}),
\end{equation}
so that
\begin{equation}
\lim_{v \to \infty} \phi(v) = 0^-.
\end{equation}
We conclude that there must exist at least one root of $\phi$ larger than $x_1$.
\end{proof}

\subsection{Some properties of the Hessian matrix}

In this section, we consider vectors $x$ so 
that their $(t)$-norm can be computed from the a.c.
part of $\mu^{\boxplus1/t}_x$, i.e. vectors $x$ so that $
m_x/k+t<1$. At such points $\normt{\cdot}$ is differentiable (and in fact $C^\infty$).
In particular, when $t<1/k$, the statements below hold true for
all $x\in\Delta_k \setminus \mathbb R 1^k$.

\begin{proposition}\label{prop:properties-H}
Let $H = H(x) = \nabla^2\|x\|_{(t)}$ be the Hessian matrix of the $(t)$-norm, taken at a point
$x \in \Delta_k^\downarrow$. Then $H$ has the following remarkable properties:
\begin{enumerate}
\item $H(x)x = H(x) {1^k} = 0$, where ${1^k} = (1 1 \cdots 1)$.
\item If $x$ is a two-valued vector, $x = (a a\cdots a b b \cdots b)$, then $H(x)$ has a block structure, i.e.
$H(x)_{ij} = 0$, whenever $i \leq m_x$ and $j > m_x$.
\item In particular, when $x=(1 0 \cdots 0)$, the first line (and column) of $H$ are null.
\item For every vector $x$, $H(x)_{1k} \geq 0$, with equality iff $x$ is constant or two-valued.
\end{enumerate}
\end{proposition}
\begin{proof}
The fact that $x$ is in the null space of the Hessian is a consequence of the homogeneity of the $(t)$-norm (and it is valid for any norm), while the second part of the first point follows from the relation $\|x+a{1^k}\|_{(t)} = \|x\|_{(t)} + a$, where $a>0$.

To prove the second statement, we need to do some explicit computations. For simplicity of
notation we will suppress the variable $x$ in the notations below. We should, however, recall that
we consider the evaluation(s) of the Cauchy-Stieltjes transform $G=G_{\mu_x}$ in the point
$w=w_x=w(x)$ provided to us by Proposition \ref{technical-proposition}. Let $m$ be the number
of $a$'s in $x$ and $l$ the number of $b$'s, $m+l=k$. We have
\begin{align*}
G &= \frac{1}{k}\left( m
(w-a)^{-1} + l (w-b)^{-1} \right)\\
\partial_iG=\frac{\partial G}{\partial x_i} &=
\begin{cases}
\frac{1}{k}(w-a)^{-2} &\qquad i \leq m\\
\frac{1}{k}(w-b)^{-2} &\qquad i > m
\end{cases}\\
G' = \partial_{k+1} G=\frac{\partial G}{\partial w}&= -\frac{1}{k}\left( m(w-a)^{-2} + l (w-b)^{-2} \right)\\
\partial_i G' = \partial_i \partial_{k+1} G &=
\begin{cases}
-\frac{2}{k}(w-a)^{-3} &\qquad i \leq m\\
-\frac{2}{k}(w-b)^{-3} &\qquad i > m
\end{cases}\\
\partial_i \partial_j G &=
\begin{cases}
\frac{2}{k}(w-a)^{-3} &\qquad i =j \leq m\\
\frac{2}{k}(w-b)^{-3} &\qquad i =j > m\\
0 &\qquad i \neq j
\end{cases}\\
G'' = \partial_{k+1}^2 G &= \frac{2}{k}\left( m(w-a)^{-3} + l (w-b)^{-3} \right)
\end{align*}

First, one shows that
\begin{equation*}
\partial_1 w = \frac{2\partial_{k+1}G \partial_1 G - G \partial_1 \partial_{k+1} G}{G \partial_{k+1}^2 G - 2 (\partial_{k+1} G)^2} = \frac{w-b}{m(a-b)},
\end{equation*}
and that the exact same
formulas are true for $\partial_i w$, when $i \leq m$. Then, by direct computation, we have that
\begin{equation}\label{eq:Hij}
H_{ij} = \frac{\partial_j G \partial_i \partial_{k+1} G - \partial_i \partial_j G \partial_{k+1} G +
\partial_i w \left( \partial_j G \partial_{k+1}^2 G - \partial_{k+1} \partial_j G \partial_{k+1} G \right)}{(\partial_{k+1} G)^2} = 0,
\end{equation}
whenever $i \leq m$ and $j >m$ (actually, by symmetry, it suffices to look at $i=1$ and $j=k$).

The third point follows form the first two: only the top-left corner of $H$ can be non-zero, but it actually is null because of the $H {1^k} = 0$ condition, 
according to part (1).

The fourth statement is trivial when $x$ is constant or two-valued, by the block-structure property.
In the case where $x$ is at least three-valued, we shall show that $H_{1k}>0$. We shall use
\begin{equation}
H_{1k} = \frac{\partial_k G \partial_1 \partial_{k+1} G - \partial_1 \partial_k G \partial_{k+1} G +
\partial_1 w \left( \partial_k G \partial_{k+1}^2 G - \partial_{k+1} \partial_k G \partial_{k+1} G \right)}{(\partial_{k+1} G)^2},
\end{equation}
with
\begin{equation}
\partial_1 w = \frac{2\partial_{k+1}G \partial_1 G - G \partial_1 \partial_{k+1} G}{G \partial_{k+1}^2 G - 2 (\partial_{k+1} G)^2}.
\end{equation}
The inequality $H_{1k}>0$ is equivalent to (all the indices run from $1$ to $k$)
\begin{equation}
\partial_1 w > (w-x_1)^{-3} \left[ \sum_i (w-x_i)^{-3} - (w-x_k)^{-1}\sum_i (w-x_i)^{-2} \right]^{-1},
\end{equation}
with
\begin{equation}
\partial_1 w = \frac{(w-x_1)^{-3}\sum_i (w-x_i)^{-1} - (w-x_1)^{-2}\sum_i (w-x_i)^{-2}}{\sum_i (w-x_i)^{-1}\sum_i (w-x_i)^{-3} - \left[ \sum_i (w-x_i)^{-2} \right]^2}.
\end{equation}
Note that $w > x_1$ and, by the Cauchy-Schwarz inequality, the denominator in the equation above is positive,
so that, after some algebraic manipulations, we obtain the following inequality
\begin{equation}
\left( 1+ \frac{w-x_1}{w-x_k} \right) \sum_i (w-x_i)^{-2} > (w-x_1)
\sum_i (w-x_i)^{-3} + (w-x_k)^{-1} \sum_i (w-x_i)^{-1}.
\end{equation}
Let us put $y_i = (w-x_i)^{-1}$, so that $y_1 \geq \cdots \geq y_k$.
The inequality becomes, after multiplying by $y_1$,
\begin{align}
&(y_1 + y_k) \sum_i y_i^2 > \sum_i y_i^3 + y_1y_k \sum_i y_i \Leftrightarrow \\
& \sum_i y_i \left( y_i^2 - (y_1+y_k) y_i + y_1y_k 
\right) < 0.
\end{align}
Note that, for each $i$, $y_i \in [y_k,y_1]$, so that
\begin{equation}
 y_i^2 - (y_1+y_k) y_i + y_1y_k 
\leq 0.
\end{equation}
Moreover, since $x$ (and thus $y$) is at least three-valued, at least one of the above inequalities
is strict, proving $H_{1k}>0$.
\end{proof}

\subsection{Local maxima of $\| \nabla \|\cdot\|_{(t)} \|_p^p$ are two-valued}

Let us first argue that all 2-valued vectors $x=(aa\cdots a b b \cdots b)$ are
critical points of the function $g(x) = \| \nabla \|x\|_{(t)} \|_p^p$ 
(when understanding the notion ``critical point'' in the usual sense of ``either zero or
non-existent derivative,'' this statement holds for all such $x\in\Delta_k$ and $t\in(0,1-1/k)$).
Recall that
differentiation in the simplex $\Delta_k$ means taking derivative in directions $y\in\mathbb
R^k$ with the property that the sum of the coordinates of $y$ is zero. Thus, let $y\in\mathbb
R^k$ be so that $\sum y_j=0$. Then
\begin{eqnarray*}
g'(x;y) & = & p\scalar{((\partial_1\normt{x})^{p-1},\dots,(\partial_n\normt{x})^{p-1})}{H(x)y}\\
& = & p\scalar{H(x)((\partial_1\normt{x})^{p-1},\dots,(\partial_n\normt{x})^{p-1})}{y}.
\end{eqnarray*}
We note that $((\partial_1\normt{x})^{p-1},\dots,(\partial_n\normt{x})^{p-1})$ is two-valued,
so that, by items (1) 
and (2) of Proposition \ref{prop:properties-H},
$$H(x)((\partial_1\normt{x})^{p-1},\dots,(\partial_n\normt{x})^{p-1})=H(x)x=0.$$

 We need to show that these are the only 
points in which the derivative of $g$ vanishes.
In fact, we will prove a bit more: we will show that in any point $x$ of differentiability for
$\normt{\cdot}$ which is not two-valued we can find a direction of ascent for $g$ inside
$\Delta_k$, thus guaranteeing that such a point is not a global maximum for $g$.

Let $x$ be at least 3-valued. Since $g$ is constant on the rays starting from
${1^k}/k$, 
 we can assume
that $x_k = 0$. We shall prove any such $x$ is not a local maximum, by exhibiting a direction of
ascent 
 $y$. First, to fix notation, let $m$ and $l$ be such that
\begin{equation}
x = (\underbrace{x_1 = x_2 = \cdots = x_m}_{m \text{ times}}> x_{m+1} \geq \cdots \geq x_{k-l} >
\underbrace{0 = \cdots  =  0}_{l \text{ times}}).
\end{equation}
With this notation, $x$ belongs to a face of co-dimension $l$ of the simplex $\Delta_k$ and we have $m+l < k$ (otherwise $x$ would be constant or two-valued).

Let us consider the direction
\begin{eqnarray}
y & = & \left( x_1 - \frac{1}{k-l}, x_2 - \frac{1}{k-l}, \cdots ,x_{k-l} - \frac{1}{k-l}, \underbrace{0, \cdots, 0}_{l \text{ times}} \right)\nonumber\\
\label{eq:direction-y}& =& x - \frac{1}{k-l} (\underbrace{1,\ldots, 1}_{k-l \text{ times}},\underbrace{0,\ldots, 0}_{l \text{ times}}),
\end{eqnarray}
which corresponds to $x$ moving away from the barycenter of the face it belongs to.
An important
feature of our choice is that, for $\varepsilon > 0$ small enough, we have that
\begin{equation}
{x} + \varepsilon y \in \Delta_k^\downarrow,
\end{equation}
so we do not leave the Weyl chamber of the simplex by moving infinitesimally in the direction $y$.
The main result of this section is the following theorem, establishing that $y$ is indeed a direction of ascent.

\begin{theorem}\label{thm:no-3-valued-max}
The direction $y$ is an ascent 
direction for $g$ at the point $x$, in the sense that
\begin{equation}
g'(x;y) =  \langle \nabla g(x) , y \rangle >0.
\end{equation}
\end{theorem}
\begin{proof}
If we set
\begin{equation}
z = H(x)y,
\end{equation}
our goal is to show that
\begin{equation}\label{eq:goal}
p^{-1}g'(x;y) = \scalar{H(x)y}{((\partial_1 \|x\|_{(t)})^{p-1},\dots,(\partial_k \|x\|_{(t)})^{p-1})} >0
\end{equation}

Note that, using Proposition \ref{prop3.4} and direct computations,
\begin{equation}\label{twentynine}
\partial_i \|x\|_{(t)} = \frac{(w-x_i)^{-2}}{-\partial_{k+1}G(w)},
\end{equation}
where the denominator above is positive. 
Moreover, given the form of the direction $y$ we have chosen \eqref{eq:direction-y} and using the fact that $H(x)1^k = H(x)x = 0$ (see Proposition \ref{prop:properties-H}), one has
$$z=H(x)y=\frac{1}{k-l}H(x)(\underbrace{0,\ldots,0}_{k-l \text{ times}}, \underbrace{1,\ldots,1}_{l \text{ times}}),$$
and thus, for all $i$,
$$z_i = \frac{1}{k-l} \sum_{j=k-l+1}^k H(x)_{ij}.$$
It follows now that equation \eqref{eq:goal} is equivalent to 
\begin{equation}\label{eq:inequality-double-sum}
\sum_{i=1}^k \sum_{j=k-l+1}^k\frac{H(x)_{ij}}{(w-x_i)^{2(p-1)}} > 0.
\end{equation}
We claim that $\sum_{i=1}^k H(x)_{ij}(w-x_i)^{-2(p-1)} $ does not depend on the actual value of $j$ between $k-l+1$ and $k$. Indeed, the first $k-l$ elements are identical for the last $l$ columns of $H(x)$, since the last $l$ components of $x$ are all zero, see equation \eqref{eq:Hij}. The columns of the bottom-right $l \times l$ corner of $H(x)$ are circular permutations of each other, so their scalar products with the constant vector $w^{-2(p-1)} 1^l$ are identical. It follows that the inequality \eqref{eq:inequality-double-sum} is $l$ times the following inequality (we take $j=k$), which is now our goal:
\begin{equation}\label{eq:inequality-simple-sum}
\sum_{i=1}^k \frac{H(x)_{ik}}{(w-x_i)^{2(p-1)}} > 0.
\end{equation}
With the change of variables 
$$a_i =\frac{1}{w-x_i}$$
and by putting, for $p>0$, $s_p = \sum_{i=1}^k a_i^p$, we have
\begin{align*}
G &= k^{-1}  s_1\\
\partial_i G = \frac{\partial G}{\partial x_i} &=  k^{-1} a_i^2\\
\partial_i \partial_j G &=  k^{-1} \delta_{ij} 2a_i^3\\
\partial_{k+1} G = \frac{\partial G}{\partial w} &= - k^{-1} s_2\\
\partial_i \partial_{k+1} G &= - k^{-1} 2a_i^3\\
\partial_{k+1} \partial_{k+1} G &=  k^{-1}  2s_3\\
\partial_i \|x\|_{(t)} &=  \frac{a_i^2}{s_2}\\
H(x)_{ij} = \partial_i \partial_j  \|x\|_{(t)} &=  \frac{2\delta_{ij}a_i^3}{s_2} + \frac{2a_i^2a_j^2}{s_2(s_1s_3-s_2^2)} [ -s_3 + s_2(a_i+a_j) - s_1 a_i a_j],
\end{align*}
where the next to last equality is a rewriting of \eqref{twentynine} and the last, of \eqref{eq:Hij}. Note that in the above equations, we have, using the Cauchy-Schwarz inequality, that $s_1s_3 > s_2^2$, the equality case being excluded using the fact that the vector $x$ (and thus $a$) is not constant. With the new notation, \eqref{eq:inequality-simple-sum} is equivalent to the following inequality (we write $q=2(p-1)$)
$$a_k^{1+q}(s_1s_3-s_2^2) + (s_2s_{3+q} - s_3 s_{2+q}) - a_k(s_1s_{3+q}-s_2s_{2+q}) > 0,$$
which is proved in the following lemma, thus completing the proof.
\end{proof}

\begin{lemma}\label{lem:inequality}
Let $a_1 \geq \cdots \geq a_k > 0$ be real numbers, and $q>0$. Then (we write $s_p = \sum_{i=1}^k a_i^p$)
$$a_k^{1+q}(s_1s_3-s_2^2) + (s_2s_{3+q} - s_3 s_{2+q}) - a_k(s_1s_{3+q}-s_2s_{2+q}) \geq 0,$$
with equality if and only if the vector $a$ is at most two-valued.
\end{lemma}
\begin{proof}
The fact that the above expression is zero for two-valued vectors is checked by direct computation. We assume from now on that the vector $a$ is at least three valued. Using the homogeneity of the inequality in $a$, we can assume $a_k=1$, and our task is now to show that$$s_1s_3-s_2^2 + s_2s_{3+q} - s_3 s_{2+q} + s_1s_{3+q}-s_2s_{2+q}  >0,$$
whenever the vector $a$ is at least three valued and $\min a = 1$. Developing 
$$s_p s_r = \sum_{i,j=1}^k a_i^p a_j^r,$$
one notices that the contributions from $i=j$ vanish in the expression above, in such a way that we only need to show that 
$$\sum_{1 \leq i < j \leq k} f(a_i,a_j) >0,$$
where the function $f$ is defined by 
$$f(x,y) = xy^3+x^3y - 2x^2y^2 + x^2y^{3+q} + x^{3+q}y^2-x^3 y^{2+q} - x^{2+q} y^3-xy^{3+q}-x^{3+q}y + x^2y^{2+q} + x^{2+q}y^2.$$
Moreover, since the initial vector $a$ is at least three valued and we already assumed that $a_k =1$, it suffices to show that, for all $x,y>1$, $x \neq y$, $f(x,y)>0$. To start, note that
$$f(x,y)=xy(x-y)(x-y-x^{1+q}+y^{1+q} +x^{1+q}y - xy^{1+q}).$$
Let us now assume that $x>y$ and show that, for all $1<y<x$, the function $g:(0,\infty) \to \mathbb R$,
$$g(q) = x-y-x^{1+q}+y^{1+q} +x^{1+q}y - xy^{1+q}=x-y+(y-1)x^{1+q}-(x-1)y^{1+q}$$
is strictly positive. We compute
$$g'(q) = (y-1)x^{1+q}\ln x - (x-1)y^{1+q}\ln y,$$
and, using the fact that the function 
$$x \mapsto \frac{x \ln x}{x-1}$$
 is increasing on $(1, \infty)$, we conclude that $g'(q)>0$ for all $q>0$. Together with the fact that $g(0)=0$, this shows that $f(x,y)>0$, whenever $x,y>1$ and $x \neq y$.
\end{proof}

\subsection{Maximum of $g$ on two-valued vectors}

From Theorem \ref{thm:no-3-valued-max}, we know that 
on the set of differentiability of $\normt{\cdot}$, all local maxima of the function $g$ are at most
two-valued.

\begin{proposition}\label{prop:3.8}
For any $ p > 1$ and $t\in(0,1-1/k)$, the maximum of the quantities
$\|x\|_p^p$ on the set of two-valued vectors $\{x\in K_{k,t}\colon\#\{{x_j}\colon
1\le j\le k\}=2\}$ is reached at $x_t^*=\nabla\normt{(1,0,\dots,0)}=(\varphi(1/k,t),
(1-\varphi(1/k,t))/(k-1),\dots,(1-\varphi(1/k,t))/(k-1))$,
where the function $\phi$ was defined in 
Proposition \ref{technical-proposition}.
\end{proposition}
\begin{proof}
Without loss of generality we restrict ourselves to 
$K_{k,t}\cap\Delta_k^\downarrow.$ Generally, the condition for a vector $(a,\dots,a,b,\dots,b)$
to belong to $K_{k,t}\cap\Delta_k^\downarrow$ is that $a\ge b$, $ma+(k-m)b=1$ and
$\scalar{\lambda}{(a,\dots,a,b,\dots,b)}\leq\normt{\lambda}$ for all $\lambda\in\Delta_k$,
where $m$ is the number of occurrences of $a$. In particular, for $\lambda=m^{-1}(1^m0^{k-m})
$ (notation from item (5) of Proposition \ref{technical-proposition}), it is necessary that we
have $a\leq m^{-1}\varphi(m/k,t)$. We note that the formula provided by (5), Proposition
\ref{technical-proposition}, equivalent to $\varphi(x,t)=\left(\sqrt{x(1-t)}+\sqrt{t(1-x)}\right)^2$,
$x=m/k$, is well defined -- and in fact an algebraic function -- for any $t\in(0,1)$ and $x\in
(0,1)$, not only on our domain $t\in(0,1-1/k)$, $x\in(0,1-t)$. Thus our proposition is proved if
we show that the function
\begin{equation}\label{2-normp}
x\mapsto x^{1-p}\varphi(x,t)^p+(1-x)^{1-p}(1-\varphi(x,t))^p
\end{equation}
is {\em decreasing} as a function of $x\in(0,1-t)$, for $p,t$ fixed as above. Indeed, this amounts to showing that the $\ell^p$ norms of probability vectors of the type $(a,\dots,a,b,\dots,b)$ with $a = m^{-1}\varphi(m/k,t)$ are maximized when $m=1$. We will prove this in
two steps.

First, let us investigate the aspect of its derivative:
\begin{eqnarray*}
\lefteqn{(1-p)\left[\left[\frac{\varphi(x,t)}{x}\right]^p-\left[\frac{1-\varphi(x,t)}{1-x}
\right]^p\right]}\\
& & \mbox{}+p\left[\left[\frac{\varphi(x,t)}{x}\right]^{p-1}-\left[\frac{1-\varphi(x,t)}{1-x}
\right]^{p-1}\right]\partial_x\varphi(x,t).
\end{eqnarray*}
For this expression to be strictly less than zero, we would need that the two (equivalent)
inequalities below hold:
\begin{eqnarray}
\lefteqn{
\frac{1}{p-1}\left[\left[\frac{\varphi(x,t)}{x}\right]^{p-1}-\left[\frac{1-\varphi(x,t)}{1-x}
\right]^{p-1}\right]\partial_x\varphi(x,t)}\label{one}\\
&<& \frac1p\left[\left[\frac{\varphi(x,t)}{x}\right]^p-\left[\frac{1-\varphi(x,t)}{1-x}
\right]^p\right],\nonumber
\end{eqnarray}
\begin{equation}\label{two}
\partial_x\varphi(x,t)<\frac{p-1}{p}
\frac{\left[\frac{\varphi(x,t)}{x}\right]^p-\left[\frac{1-\varphi(x,t)}{1-x}\right]^p
}{\left[\frac{\varphi(x,t)}{x}\right]^{p-1}-\left[\frac{1-\varphi(x,t)}{1-x}
\right]^{p-1}
}.
\end{equation}
Our strategy is to first show that the map $p\mapsto\frac{p-1}{p}\frac{a^p-b^p}{a^{p-1}-b^{p-1}}
$ for fixed
$$
0<b=\frac{1-\varphi(x,t)}{1-x}<1<a=\frac{\varphi(x,t)}{x}
$$
is increasing on $[1,+\infty)$, and then show that inequality \eqref{one}
holds when we take
$p\searrow1$. Continuity in $p$ will then provide the desired result.
 Since $\varphi(x,t)>x$, we indeed have
$0<b<1<a.$

 Let us now prove the
first step. For simplicity, we shall let $c=a/b>1$ and then $
\frac{p-1}{p}\frac{a^p-b^p}{a^{p-1}-b^{p-1}}=\frac{a}{c}\cdot\frac{c^p-1}{p}\cdot
\frac{p-1}{c^{p-1}-1}.$ Thus, it will be enough to show that
$$
f_c\colon[1,+\infty)\to[0,+\infty),\quad f_c(p)=\frac{c^p-1}{p}\cdot
\frac{p-1}{c^{p-1}-1}
$$
is increasing. We re-write $f_c$ as $f_c(p)=\left(1-\frac1p\right)c+\left(1-\frac1p\right)
\frac{c-1}{c^{p-1}-1}$.
Then $\partial_pf_c(p)=\frac{1}{p^2}c+\frac{1}{p^2}\frac{c-1}{c^{p-1}-1}-
\left(1-\frac1p\right)\frac{(c-1)c^{p-1}\log c}{(c^{p-1}-1)^2}$. Clearly,
$\partial_pf_c(1^+)=+\infty$,
while $\partial_pf_c(+\infty)=0$. Thus, close to 1, $f_c$ is indeed necessarily increasing, 
regardless of $c>1$. The statement $\partial_pf_c(p)\ge0$ is equivalent to
$\frac{(c^p-1)(c^{p-1}-1)}{c^{p-1}(c-1)}-p(p-1)\log(c)\ge0$. We denote the left hand side
by $h(c)$ (since we shall analyse here the dependence on $c$, we suppress from the notation the
dependence in $p$). First note that $h(1^+)=0$. We have
$$
h'(c)=\frac{(c^p-cp+p-1)((p-1)c^{p}-pc^{p-1}+1)}{c^{p}(c-1)^2}.
$$
As all factors in this expression are trivially positive when $p,c\in(1,+\infty)$, so is $h'(c)$.
This completes the first step.

Note that inequality \eqref{one} when $p\searrow 1$ becomes simply
$$
\left(\log\left[\frac{\varphi(x,t)}{x}\right]-\log\left[\frac{1-\varphi(x,t)}{1-x}\right]\right)
\partial_x\varphi(x,t)\leq\frac{\varphi(x,t)}{x}-\frac{1-\varphi(x,t)}{1-x}.
$$
It will be convenient to divide by $\partial_x\varphi(x,t)$ in the above and move all terms
to the right before differentiating in $x$ in order to find the point of minimum for
this expression and find it to be nonnegative:
\begin{equation}\label{sh}
x\mapsto
\frac{\frac{\varphi(x,t)}{x}-\frac{1-\varphi(x,t)}{1-x}}{\partial_x\varphi(x,t)}
-\left(\log\left[\frac{\varphi(x,t)}{x}\right]-\log\left[\frac{1-\varphi(x,t)}{1-x}\right]\right).
\end{equation}

The expression of this derivative in $x$ is too cumbersome to be provided, but
the change of variable $t=\cos^2(r)$ and $x=\sin^2(s)$, where we will allow $r$ to vary in $(0,\pi
/2)$ and $s$ in ($0,r)$ allows a simplification. It can
be shown that (see \cite{num} for the details) with these variables, this derivative is
$$
-16\frac{\cos^3(r)\sin(r-2s)}{\sin^3(2s)\sin^2(2(r-s))},
$$
which cancels only at $s=\frac{r}{2}$.
Our function from equation 
\eqref{sh} becomes
$$
4\frac{\cos(r)\cos(r-2s)}{\sin(2r-2s)\sin(2s)}+2\log\left(\tan(r-s)\tan(s)\right).
$$
The value at the critical point is
$4\frac{\cos(r)}{\sin^2(r)}+4\log(\tan(r/2))$, positive whenever $r\in(0,\pi/2)$. Indeed,
its derivative as a function of $r$ is \[\displaystyle - \frac{1}{\sin \left( r \right)}-2\,{\frac { \cos^2\left( r \right)}{ \sin^3\left( r \right)}}+ \left( \frac12+\,{\frac { \sin^2 \left( r/2 \right)}{2 \cos^2 \left( r/2 \right) }} \right) \frac{\cos \left( r/2\right)}{\sin \left( r/2 \right) }\]
or, in a nicer form,
$$
-\frac{1}{\sin \left( r \right) }-2\,{\frac { \cos^2 \left( r \right)}{\sin^3 \left( r \right)}}+\frac{1}{\sin(r)}=-2\,{\frac { \cos^2 \left( r \right)}{\sin^3 \left( r \right)}},
$$
which is obviously negative.

This way we have proved the positivity of the function in \eqref{sh} for all $t\in(0,1)$ and
$x\in(0,1-t)$, which concludes our proof.
\end{proof}

\begin{corollary}\label{thm:g-global-maximum}
When $t\le 1/k$,
the global maximum of the function $g(x)=\|\nabla\normt{x}\|_p^p$, $x\in\Delta_k$,
is attained at the point $e_1 = (1, 0, \ldots, 0)$. We have
\begin{equation}
g(e_1) =\|x^*_t\|_p^p=\max_{\lambda \in K_{k,t}} \|\lambda\|_p^p =
 \phi(1/k,t)^p + (k-1)\left[1-\phi(1/k,t)\right]^p,
\end{equation}
where the function $\phi$ was defined in 
Proposition \ref{technical-proposition}.
\end{corollary}
\subsection{The general $t$ case}

In the previous sections we have proved the Theorem \ref{thm:minimum-pnorm-Kkt}
in the case $t\leq 1/k$.
Now we prove it in full generality.

The case $t\geq 1-1/k$ is trivial (in this case the $(t)$-norm is the operator norm and $\xopt = e_1$), so we focus on
the case where $t\in (1/k,1-1/k)$. We will require the following well-known notions and
results (see \cite[Sections 18 and 25]{rockafellar}).  Given a convex set $C$ in an Euclidean space
$\mathbb R^k$ and a point $x\in C$, a {\em supporting hyperplane} of $C$ at $x$ is a
$k-1$-dimensional affine manifold in $\mathbb R^k$ which contains $x$ and so that $C$
is included entirely in exactly one of the two closed half-spaces determined by this manifold.
An {\em exposed point} of $C$ is a point through which there is a supporting hyperplane
of $C$ which contains no other point of $C$.
\begin{theorem}[Straszewicz]
For any closed convex set $C$, the set of exposed points of $C$ is a dense subset of the set
of extreme points of $C$.
\end{theorem}
The set of exposed points of a polar dual set is characterized by
\cite[Corollary 25.1.3]{rockafellar}. We will apply this result to $K_{k,t}$:
\begin{proposition}
For all $k\in\mathbb N$, $t\in(0,1)$, the set of exposed points of $K_{k,t}$ coincides
with the image of the points of differentiability of $\normt{\cdot}$
$$
\{\nabla\normt{x}\colon x\in\Delta_k, \normt{\cdot}\textrm{ differentiable at }x\}.
$$
\end{proposition}
We can now complete the proof of our main theorem:

\begin{proof}[Proof of Theorem \ref{thm:minimum-pnorm-Kkt}]
The function $\|\cdot\|_p^p$ being convex, its maximum is reached on an extremal point of $K_{k,t}$.
Therefore, by the above proposition,
$$
\max\{\|\lambda\|_p^p\colon\lambda\in K_{k,t}\}=\sup\{\|\nabla\normt{x}\|_p^p\colon
x\in\Delta_k,\normt{\cdot}\textrm{ differentiable at }x\}.
$$
By Theorem \ref{thm:no-3-valued-max} and Proposition \ref{prop:3.8}, among points of
differentiability, the maximum of  $\|\nabla\normt{x}\|_p^p$ is reached at $e_1$.
Since $x^*_t=\|\nabla\normt{e_1}\|_p^p$, this
concludes the proof of Theorem \ref{thm:minimum-pnorm-Kkt} in full generality.
\end{proof}

\section{Minimum output entropy for quantum channels}
\label{sec:MOE}

In the reminder of the paper, we apply the minimization result of Theorem \ref{thm:minimum-pnorm-Kkt} 
to the problem of the \emph{minimum output entropy} of quantum channels.

Quantum channels \cite{nch} are linear, completely positive and trace preserving maps which model the 
most general evolution of quantum systems. In Quantum Information Theory, they are used to model 
information transmission, and several notions of channel capacities have been introduced. In what follows, 
we are interested in the \emph{classical capacity} of channels, a measure of how fast classical information 
can be transmitted with the help of quantum channels.

A quantum channel $\Phi: M_d(\mathbb C) \to M_k(\mathbb C)$ is a linear map which has the following two properties:
\begin{itemize}
\item trace preservation: $\forall X \in M_d(\mathbb C)$, $\mathrm{Tr} \Phi(X) = \mathrm{Tr} X$;
\item complete positivity: $\forall s \geq 1$, the map $\Phi \otimes \mathrm{id}_s$ is positive.
\end{itemize}

The information transmission capacity of such a channel is characterized by its classical information, 
$C(\Phi)$, which measures, asymptotically, how many uses of the channel are required to send one bit of 
classical information. Computing the classical capacity of quantum channels \cite{hol, swe} is a difficult 
problem whereas the capacity of classical channels (Markov maps) was computed by Shannon in his seminal 
paper \cite{sha}. The main difficulty in the quantum setting is the need of regularization,
\begin{equation}\label{eq:capacity}
C(\Phi) = \lim_{r \to \infty} \frac{1}{r}\chi(\Phi^{\otimes r}),
\end{equation}
where the quantity $\chi$ is the so-called \emph{Holevo capacity} (or the one-shot capacity) \cite{hol} of the channel,
\begin{equation}
\chi(\Phi) = \max_{(p_i),(\rho_i)} H(\sum_i p_i \Phi(\rho_i)) - \sum_i p_i H(\Phi(\rho_i)),
\end{equation}
the maximum being taken over probability vectors $(p_i)$, $p_i \geq 0$, $\sum_i p_i = 1$ and quantum states 
$(\rho_i)$, $\rho_i \in M_d(\mathbb C)$, $\rho_i \geq 0$, $\mathrm{Tr} \rho_i =1$. The function $H$ denotes 
the \emph{von Neumann entropy}, the extension (by functional calculus) of the Shannon entropy to quantum states
\begin{equation}
H(\rho) = - \mathrm{Tr}(\rho \log \rho).
\end{equation}

For some time, the Holevo quantity $\chi$ was conjectured to be additive, in the sense that for all quantum 
channels $\Phi, \Psi$,
\begin{equation}
\chi(\Phi \otimes \Psi) = \chi(\Phi) + \chi(\Psi).
\end{equation}
If such an additivity property would hold, there would be no need for the regularization procedure in 
equation \eqref{eq:capacity} and the classical capacity of $\Phi$ would be equal to its one-shot capacity. 
Shor showed \cite{sho} that the additivity of $\chi$ is equivalent to similar properties of other quantities of 
interest in quantum information, the foremost being the minimum output entropy of channels \cite{kru}
\begin{equation}
H^{\min}(\Phi) = \min_{\substack{\rho\in M_d(\C) \\ \rho\geq 0, \mathrm{Tr}\rho=1}} H(\Phi(\rho)).
\end{equation}
The focus of the community shifted to showing additivity for the minimum output entropy, or its $p$-variants, 
called \emph{R\'enyi entropies}. These are defined for probability vectors $x \in \Delta_k$ by
\begin{equation}
H_p(x) = \frac{1}{1-p}\log\sum_{i=1}^k x_i^p,
\end{equation}
and extended by functional calculus to quantum states $\rho$. Note that the above definitions are valid for 
$p \in (0, \infty)$, the value in $p=1$, obtained by taking a limit, coinciding with the von Neumann entropy $H$. 
The $\min$ variants are defined by
\begin{equation}
H_p^{\min}(\Phi) = \min_{\substack{\rho\in M_d(\C) \\ \rho\geq 0, \mathrm{Tr}\rho=1}} H_p(\Phi(\rho)).
\end{equation}

The additivity property for the quantities $H_p^{\min}$ was shown to be false, in a series of papers 
\cite{hwe, win, hay, chl} culminating with Hastings' counterexample \cite{has}. Since the resolution of the 
additivity conjecture, effort has been put \cite{fki, fkm, bho, cn1, cn2, cn3, cn-entropy, cfn1, cfn2, ghp, fne} 
into understanding, extending and improving the deviations from additivity.

The remainder of the paper contains two main results. The first one provides a limit value for the minimum 
$p$-output entropy of random quantum channels, while the second one deals with counterexamples to the 
additivity relation for the quantity $H_p^{\min}$.

\section{Limiting value of the minimum output entropy for large random quantum channels}\label{sec:limiting}

\subsection{Random quantum channels and the subspace model}

We shall endow the set of quantum channels $\Phi: M_d(\mathbb C) \to M_k(\mathbb C)$ with a natural 
probability measure and we shall refer to channels sampled from this measure 
as \emph{random quantum channels}.

The idea behind the model of random quantum channels we are considering (which is standard in the literature, 
see \cite{hayden-winter}) is the \emph{Stinespring dilation theorem}, which asserts that any completely 
positive, trace preserving map $\Phi$ can be realized as
\begin{equation}\label{eq:channel-from-isometry}
\Phi(X) = [\mathrm{id} \otimes \mathrm{Tr}](WXW^*),
\end{equation}
where $n$ is an integer (called the dimension of the environment) and
\begin{equation}
W : \mathbb C^d \to \mathbb C^k \otimes \mathbb C^n
\end{equation}
is an isometry, $W^*W = I_d$. Conversely, any isometry $W$ gives rise to a quantum channel.

The set of all isometries $W : \mathbb C^d \to \mathbb C^k \otimes \mathbb C^n$ 
admits a left- and right- invariant probability measure, called the Haar measure, which can be obtained, 
say, from the Haar measure on the unitary group $\mathcal U(kn)$. For each integer dimension $n$, 
we shall endow the set of all channels with the measure induced by the probability on the set of isometries 
$W$ by the map which associates to $W$ the channel \eqref{eq:channel-from-isometry}. 
Such a channel will be called a random channel with environment dimension $n$.

A crucial observation is that the minimum output entropy of a channel depends only on its output, and not on the exact way in which the input is mapped to the output. In our isometry picture, the object of interest is the output set
\begin{equation}
\{[\mathrm{id} \otimes \mathrm{Tr}](W \rho W^*) \, : \, \rho \text{ quantum state}\}.
\end{equation}
Moreover, note that the entropy functionals are convex, for all $p \geq 1$; hence, their minimum is 
attained on the extremal points of the set of states, i.e. rank-one projections 
$P_x$, $x \in \mathbb C^d$, $\|x\| = 1$. We are thus interested in the entropies of the set of quantum states
\begin{equation}
\{[\mathrm{id} \otimes \mathrm{Tr}](P_{Wx}) \, : \, x \in \mathbb C^d, \|x\| = 1\}.
\end{equation}

The eigenvalues of the partial trace $[\mathrm{id} \otimes \mathrm{Tr}](P_{y})$ are called 
the \emph{singular values} (or the \emph{Schmidt coefficients}) of the vector $y \in \C^k\otimes \C^n$: 
they are the numbers $\lambda_{1}(y)\geq\ldots \geq \lambda_{k}(y)\geq 0$ such that
\begin{equation}
        y=\sum_{i=1}^{k}\sqrt{\lambda_{i}(y)} \, e_{i}(y)\otimes f_{i}(y)
\end{equation}
where $e_{i}(y)$ (resp.~ $f_{i}(y)$) are orthonormal vectors in $\C^k$ (resp.~ $\C^n$).
If $y$ is a norm one vector in the Euclidean space $\C^{kn}$, then $\lambda (y)=(\lambda_{1}(y), \ldots,  \lambda_{k}(y))$
belongs to the set $ \Delta_k^\downarrow$.

Going back to our isometry picture for quantum channels, we notice that the image subspace
\begin{equation}
V = \mathrm{Im} W \subset \mathbb C^k \otimes \mathbb C^n, \quad \dim V = d
\end{equation}
contains all the information needed to compute minimum output entropies:
\begin{equation}
H_p^{\min}(\Phi) = \min_{x \in V, \|x\| =1} H_p(\lambda(x)).
\end{equation}

To an output subspace $V  \subset \mathbb C^k \otimes \mathbb C^n$, we associate its \emph{singular value set}
\begin{equation}
\tilde K_V = \{ \lambda(x) \, : \, x \in V, \|x\| = 1\} \subset \Delta_k^\downarrow.
\end{equation}

The image measure of the Haar probability measure on the set of isometries through the map 
$W \mapsto V = \mathrm{Im} W$ is the Haar measure on the Grassmann manifold  
$\Gr_d(\C^k\otimes \C^n)$ of subspaces of $\C^k\otimes \C^n$ with dimension $d$. In this way, 
$\tilde K_V$ is a random subset of $\Delta_k^\downarrow$. For technical reasons, it will be convenient to 
replace it by
\begin{equation}
K_V = \{ (\lambda_{\sigma(1)}, \lambda_{\sigma(2)}, \ldots, \lambda_{\sigma(k)}) \, : \, \lambda \in \tilde K_V \text { and } \sigma \in S_k\} \subset \Delta_k,
\end{equation}
which is its \emph{symmetrized version} under permuting the coordinates.

\subsection{The large $n$ asymptotics}\label{sec:large-asymptotics}

We are interested in a random 
sequence $V_{n}$ of subspaces of $\C^k \otimes \C^n$ having the following properties:
\begin{enumerate}
\item $V_{n}$ has dimension $d_n$ which satisfies $d_n \sim tkn$;
\item The law of $V_{n}$ follows the invariant measure on the Grassmann manifold $\Gr_{d_n}(\C^{k}\otimes\C^n)$.
\end{enumerate}

In this setting, we call $K_{n,k,t}= K_{V_{n}}$. We recall the following theorem, 
which was our main theorem in \cite{bcn1}:

\begin{theorem}\label{thm:main-bcn1}
Almost surely, the following hold true:
\begin{itemize}
\item Let $\mathcal{O}$ be an open set in $\Delta_{k}$ containing $K_{k,t}$.
Then, for $n$ large enough, $K_{n,k,t}\subset \mathcal{O}$.
\item Let $\mathcal K$ be a compact set in the interior of $K_{k,t}$.
Then, for $n$ large enough, $\mathcal K \subset K_{n,k,t}$.
\end{itemize}
\end{theorem}

\subsection{Convergence result for the minimum output entropy}

Putting together Theorem \ref{thm:main-bcn1} proved in \cite{bcn1} and Theorem \ref{thm:minimum-pnorm-Kkt}
 proved in Section \ref{sec:proof}, we obtain the following convergence result 
 for the minimum output $p$-entropies of random quantum channels.

\begin{theorem}\label{thm:MOE-convergence}
Let $p$ be a real number in $[1,\infty]$ and $\Phi_n : M_{d_n}(\mathbb C) \to M_k(\mathbb C)$ a sequence of random quantum channels with constant output space of dimension $k$, environment of size $n \to \infty$ and input space of dimension $d_n \sim tkn$. Then, almost surely as $n \to \iy$,
\begin{equation}
\lim_{n \to \infty} H_p^{\min}(\Phi_n) = H_p(\xopt),
\end{equation}
with $\xopt$ defined in equation \eqref{eq:def-xopt}.
\end{theorem}

\begin{proof}
In the case $p>1$,
this follows right away from Theorems \ref{thm:main-bcn1} and \ref{thm:minimum-pnorm-Kkt}.
The case $p=1$ can be obtained by continuity of the entropy.
\end{proof}

\section{Violation of the additivity for minimum output entropies}
\label{sec:violation}

\subsection{The MOE additivity problem}

The following theorem summarizes some of the most important breakthroughs in quantum information
theory in the last decade. It is based in particular on the papers \cite{has, hayden-winter}.

\begin{theorem}
For every $p \in [1, \infty]$, there exist
quantum channels $\Phi$ and $\Psi$ such that
\begin{equation}
H_p^{\min}(\Phi \otimes \Psi) < H_p^{\min}(\Phi) + H_p^{\min}(\Psi).
\end{equation}
\end{theorem}

Except for some particular cases ($p>4.73$, \cite{hwe} and $p>2$, \cite{ghp}), the proof of this theorem uses 
the random method, i.e. the channels $\Phi, \Psi$ are random channels, and the above inequality occurs with 
non-zero probability. At this moment, we are not aware of any explicit, non-random choices for $\Phi, \Psi$ in the case  $1 \leq p \leq 2$.

Moreover, the strategy in all the results cited above are based on the \emph{Bell phenomenon}, 
i.e. the choice $\Psi  = \bar \Phi$ and the use of the maximally entangled state as an input for 
$\Phi \otimes \bar \Phi$.

\subsection{The Bell phenomenon}

In order to obtain violations for the additivity relation of the minimum output entropy, one needs 
to obtain upper bounds for the quantity $H_p^{\min}(\Phi \otimes \Psi)$. The idea of using conjugate 
channels ($\Psi  = \bar \Phi$) and bounding the minimum output entropy by the value of the entropy 
at the Bell state dates back to \cite{win}. To date, it has proven to be the most successful method of 
tackling the additivity problem. Several results show that the choice of the Bell state in the conjugate 
channel setting might not be far from optimal \cite{cfn1, fne}. The following inequality is elementary 
and lies at the heart of the method
\begin{equation}
H_p^{\min} (\Phi \otimes \bar \Phi) \leq H^p_{\min} ([ \Phi \otimes \bar \Phi] (E_{d})),
\end{equation}
where $E_{d}$ is the maximally entangled state over the input space $(\C^d)^{\otimes 2}$. 
More precisely, $E_d$ is the projection on the Bell vector
\begin{equation}
\Bell_d = \frac{1}{\sqrt{d}}\sum_{i=1}^{d} e_i \otimes e_i,
\end{equation}
where $\{e_i\}_{i=1}^{d}$ is a fixed basis of $\C^{d}$.

For random quantum channels $\Phi = \Phi_n$, the random output matrix $[\Phi_n\otimes \bar \Phi_n] (E_{d})$
was thoroughly studied in \cite{cn1} in the regime $d \sim tkn$; we recall here one of the main results of that paper.

\begin{theorem}\label{thm:bi-canal}
Almost surely, as $n$ tends to infinity, the random matrix $[\Phi_n\otimes\bar \Phi_n](E_{tkn}) \in M_{k^2}(\C)$ 
has eigenvalues
\begin{equation}\label{eq:def-gamma-opt}
\gamma^*_t = \left( t + \frac{1-t}{k^2},\underbrace{\frac{1-t}{k^2}, \ldots, \frac{1-t}{k^2}}_{k^2-1 \text{ times}}\right).
\end{equation}
\end{theorem}

This result improves on a bound  \cite{hayden-winter} via linear algebra techniques,
which states that the largest eigenvalue of the random matrix $[\Phi_n\otimes \bar \Phi_n](E_{tkn})$
 is at least $t$. The improvement 
provided by Theorem \ref{thm:bi-canal} comes from the fact that the largest 
 eigenvalue of the output is larger (by $(1-t)k^{-2}$). In the next section, we will show how this improvement 
 leads to better bounds for the size of channels which exhibit violations.

\subsection{Macroscopic violations for the minimum output entropy of random quantum channels}

In this section, we fix $p=1$, so we shall study the most important case of Shannon - von Neumann entropy. 
The main theorem of this section was the initial motivation for the line of work started in \cite{bcn1}: 
we want to obtain \emph{large} violations for the additivity relation, for \emph{reasonable values} 
of the model parameter $k$. Note that previous work showed that violations of size $\approx 10^{-6}$ 
exist for channels with output space of dimension  $\approx 10^4$ \cite[Proposition 3]{fkm}. 
We drastically improve these results with the following result.

\begin{theorem}\label{thm:violation}
For any output dimension $k\geq 183$, in the limit $n \to \iy$, there exist values of the parameter 
$t$ such that almost all
random quantum channels violate the additivity of the von Neumann minimum output entropy. 
For large enough values of $k$, the violation can be made as close as desired to 1 bit.

Moreover, in the same asymptotic regime, for all $k < 183$, the  von Neumann entropy of the output state
$[\Phi_n\otimes \bar \Phi_n](E_{tnk})$ is almost surely larger than $2 H^{\min}(\Phi_n)$. 
Hence, in this case, one can not exhibit violations of the additivity using the Bell state as an input
for the product of conjugate random quantum channels.
\end{theorem}

\begin{proof}
The result follows from an analysis of the entropies of the two probability vectors 
$\xopt$ and $\gamma^*_t$ from Theorems \ref{thm:MOE-convergence} and \ref{thm:bi-canal}.
We estimate the following almost sure asymptotic entropy difference:
\begin{equation}
D(k, t) = \lim_{n \to \iy} H\left([\Phi_n\otimes \bar \Phi_n](E_{tnk})\right) - 2H^{\min}(\Phi_n),
\end{equation}
which is an upper bound for $H^{\min}(\Phi_n\otimes \bar \Phi_n) - 2H^{\min}(\Phi_n)$. Using
Theorems \ref{thm:minimum-pnorm-Kkt} and \ref{thm:bi-canal}, we have that 
$D(k, t) = H(\gamma^*_t) - 2H(\xopt)$, for which an analytic expression is available, from equations 
\eqref{eq:def-xopt} and \eqref{eq:def-gamma-opt}.
A numerical study \cite{num} of this function (see Figure \ref{fig:ent-diff}) shows that $D(k, t) > 0$ for all
$k < 183$ and all $t \in (0, 1)$. Violations (i.e. negative values for $D(k, t)$) appear for the first 
time at $k=183$ and $t \approx 0.11$.

\begin{figure}[htbp]
\centering
\subfigure[]{\includegraphics[width=0.43\textwidth, bb = 30 25 250 250]{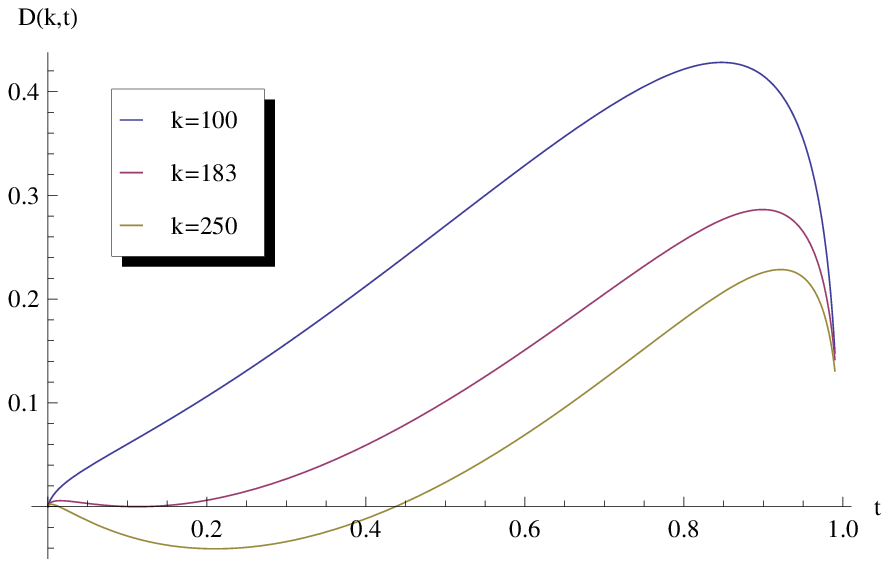}}\qquad
\subfigure[]{\includegraphics[width=0.43\textwidth, bb = 30 25 250 250]{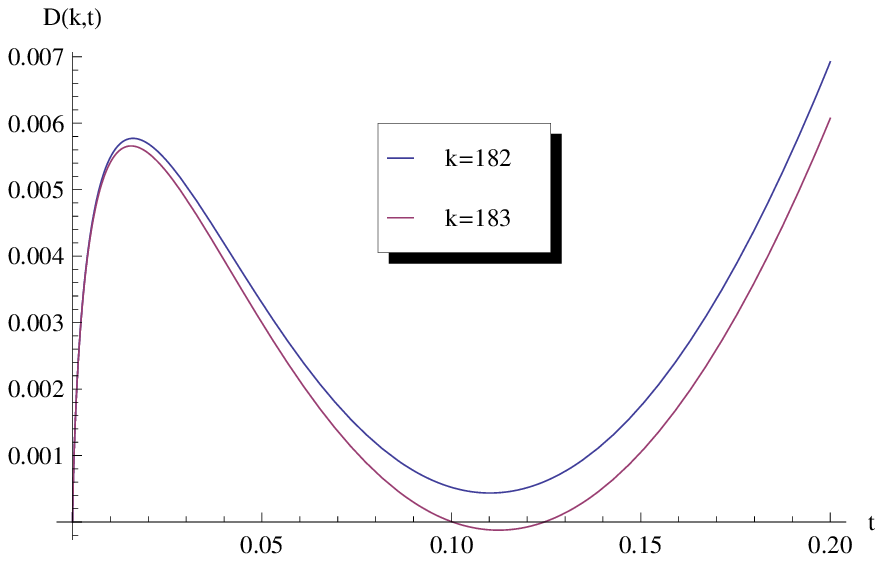}}\\
\caption{The entropy difference function $D(k,t)$. On the left side, the plots for $k=100, 183, 250$ 
(top to bottom) and $t \in (0,1)$ are given. An enlargement of the area where violations appear for the 
first time is presented in (b), for $k=182, 183$.}
\label{fig:ent-diff}
\end{figure}

An asymptotical expansion of the explicit function $D(k, t)$ at fixed $t$ and $k \to \iy$ shows that
$$D(k, t) = t \log t  + (1-t) \log(1-t) + o(1) = -H(t, 1-t) + o(1).$$
This shows that, for $t=1/2$, the quantity $D(k,1/2)$ is negative, for $k$ large enough. Analysis of the
function $k\to D(k,1/2)$ shows that it is negative iff 
$k\geq 276$, which implies that there is a violation of additivity 
for $k\geq 276$.
A numerical study of $t\to D(k,t)$
for all $k\in \{183,\ldots ,276\}$
allows to conclude that violations are observed iff $k\geq 183$, proving one of the claims of the theorem.

Moreover, the maximal violation of $\log 2$ (1 bit),  is achieved for $t=1/2$ and very large values of $k$. Note that the parameter value $t=1/2$ has been already used in \cite{cn2} to obtain violations of $p$-R\'enyi entropy additivity for $p>1$.
\end{proof}

Several remarks and comments about the theorem are in order now. 
\begin{enumerate}
\item Let us point out the improvement we obtained over previous results for the size of the violation. For the first time, \emph{macroscopic} violations are obtained for the minimum output entropy; in particular, the size of the violation increases with output size. 
\item In contrast with the case of $p$-R\'enyi entropies for $p>1$ where violations of size $\log k$ have been obtained in \cite{hayden-winter, cn2}, for the von Neumann entropy we get \emph{bounded} violations, of order $\log 2$, which do not grow to infinity with $k$. 
\item Also, the smallest output dimension for which violations are observed is $k=183$, which corresponds to approximately 8 qubits.
\end{enumerate}

The large asymptotic violation, which can be made arbitrarily close to 1 bit, is achieved for $t=1/2$ and large $k$. In Figure \ref{fig:ent-diff-2-k}, we have plotted the function $D(k,1/2)$, fixing $t=1/2$. One can then observe that the first violation (negative value of $D$) appears at $k=276$, see \cite{num}. Another interesting choice is $t=1/k$, which corresponds to channels with equal input and output spaces. The plot of the entropy difference in this regime can be found in Figure \ref{fig:ent-diff-2-k}. A numerical analysis shows that the first violation appears at $k=432$. Note that in this case, one can also write a series expansion for $D$:
\begin{equation}
D(k, 1/k) = -\frac{\log k}{k} + o(k^{-1}\log k),
\end{equation}
which agrees with the vanishing violation observed in \cite{has}.

\begin{figure}[htbp]
\centering
\subfigure[]{\includegraphics[width=0.43\textwidth]{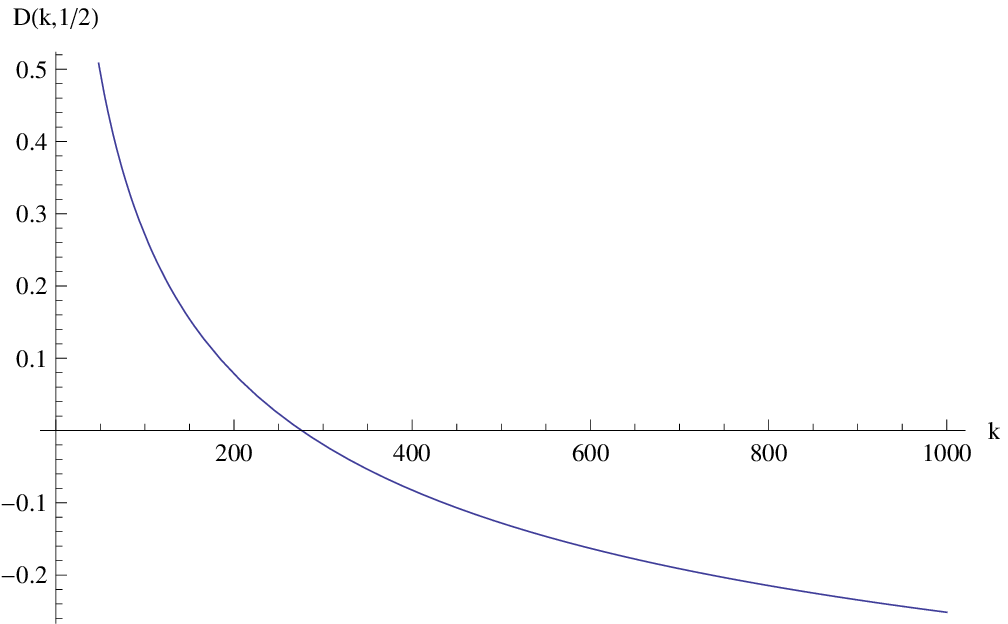}}\qquad
\subfigure[]{\includegraphics[width=0.43\textwidth]{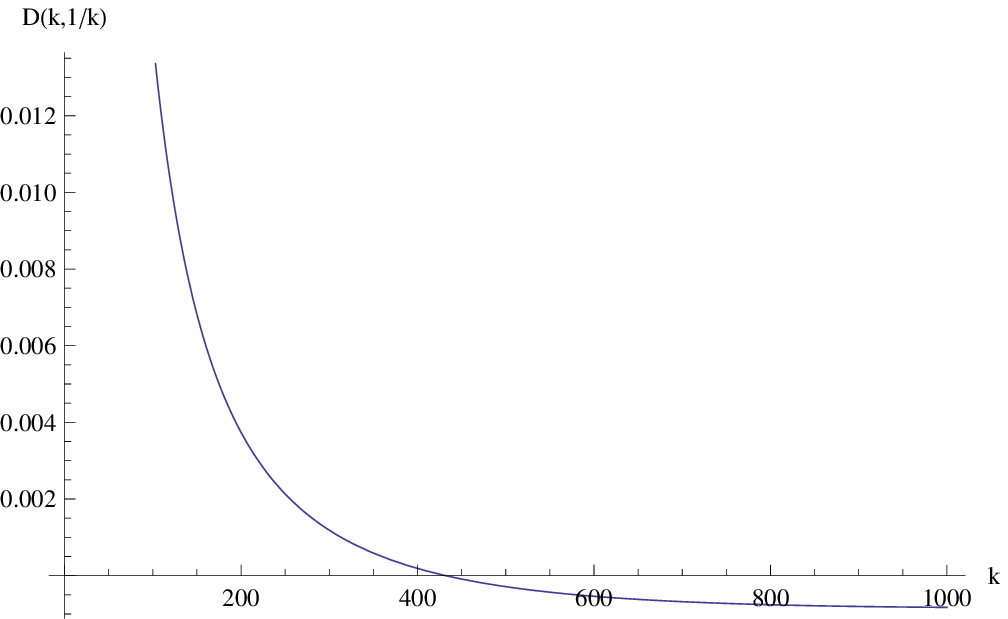}}\\
\caption{The entropy difference function $D(k,t)$ as a function of $k$, for $t=1/2$ and $t=1/k$ respectively.}
\label{fig:ent-diff-2-k}
\end{figure}

We would also like to point out that our result on $p$-norm maximization on $K_{k,t}$ implies, after a numerical study similar to the one used in the proof of Theorem \ref{thm:violation}, that violations for the $p$-R\'enyi entropy are observed for the first time at $k=16$ for $p=2$, $k=14$ for $p=3$ and $k=13$ for $p\in [4, \infty]$, see \cite{num}.

Next, we would like to emphasize the importance of
Theorem \ref{thm:bi-canal} derived in \cite{cn1}. Without this result, one has to rely on
the Hayden-Winter bound \cite{hayden-winter} and replace the output eigenvalue vector $\gamma^*_t$ from Theorem \ref{thm:bi-canal} with the more mixed vector
\begin{equation}\label{eq:def-gamma-opt-HW}
\gamma^{HW}_t = \left( t, \underbrace{\frac{1-t}{k^2-1}, \ldots, \frac{1-t}{k^2-1}}_{k^2-1 \text{ times}}\right).
\end{equation}
This leads to a larger entropy difference
\begin{equation}
D^{HW}(k, t) = H_p(\gamma^{HW}_t ) - 2H_p(\xopt).
\end{equation}

A numerical analysis of this problem, presented in Figure \ref{fig:ent-diff-HW}, shows that the first violations
appear for $k=184$. The use of the exact result from \cite{cn1} improves thus by one the bound on 
the minimum size of channels which exhibit violations of the
additivity of the minimum output entropy. Note however that one can still achieve values of the violation 
arbitrarily close to 1 bit using the Hayden-Winter bound \cite{hayden-winter}.

\begin{figure}[htbp]
\centering
\subfigure[]{\includegraphics[width=0.43\textwidth, bb = 30 25 250 250]{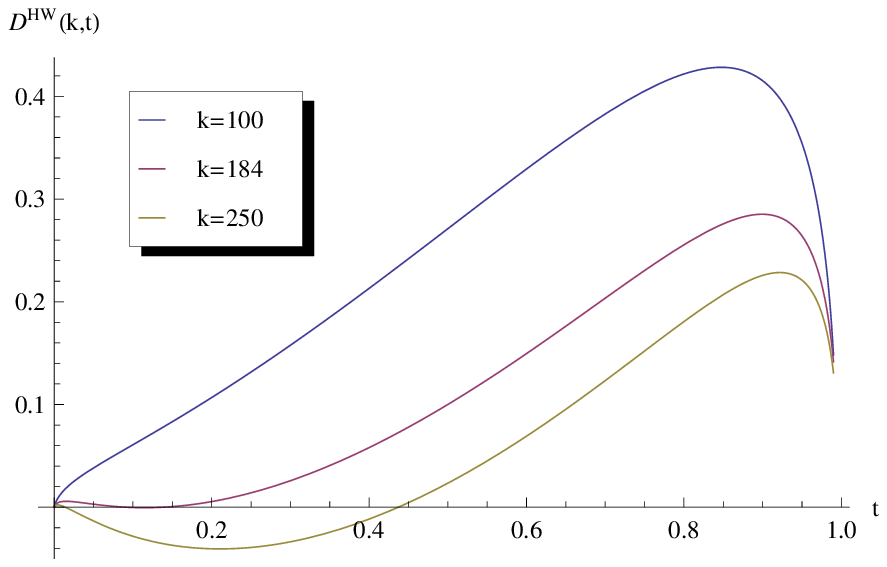}}\qquad
\subfigure[]{\includegraphics[width=0.43\textwidth, bb = 30 25 250 250]{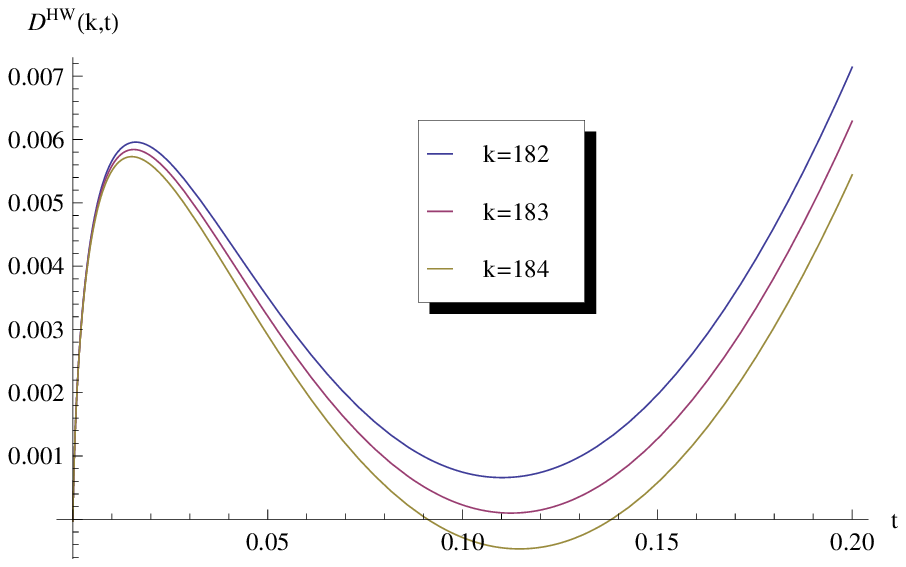}}\\
\caption{The entropy difference function $D^{HW}(k,t)$. On the left side, the plots for $k=100, 184, 250$ (top to bottom) and $t \in (0,1)$ are given. An enlargement of the area where violations appear for the fist time is presented in (b), for $k=182, 183, 184$.}
\label{fig:ent-diff-HW}
\end{figure}

Our result does not imply that, almost surely, there is  no violation of the additivity of the 
minimum output entropy for $k<183$. What we prove is that the Bell state will not yield such counterexamples. 
Some other input state for the product channel might provide better upper bounds. 
Work in this direction \cite{cfn1, fne} shows however that the Bell state is not far from being the 
optimal input state for product of conjugate random quantum channels. 
We conjecture thus that the violation of $1$ bit is indeed the maximal one in the current setting.

As a final remark, note that our techniques do not provide any information on the size of the environment 
dimension $n$. We plan to address this question in a subsequent paper, since the techniques required to 
tackle bounds on the environment dimension are of very different nature.

\end{document}